\documentclass{amsart} 

\usepackage{natbib}
\usepackage{graphicx}
\usepackage{amssymb,amsmath}

\newtheorem{theorem}{Theorem}[section]
\newtheorem{lemma}[theorem]{Lemma}
\newtheorem{proposition}[theorem]{Proposition}
\newtheorem{corollary}[theorem]{Corollary}

\theoremstyle{definition}
\newtheorem{definition}[theorem]{Definition}

\theoremstyle{remark}

\newcommand{\RR}{{\mathbb R}}

\begin{document}
\title{Phylogenetic trees and Euclidean embeddings}

\author{ Mark Layer}
\author{John A. Rhodes }
\address{Department of Mathematics and Statistics, University of Alaska Fairbanks, Fairbanks, AK, 99775, USA}
\email{j.rhodes@alaska.edu}


\begin{abstract}
It was recently observed by de Vienne et al. that  a simple square root transformation of distances between taxa on a phylogenetic tree allowed for an embedding of the taxa into Euclidean space. While the justification for this was based on a diffusion model of continuous character evolution along the tree, here we give a direct and elementary explanation for it that provides substantial additional insight. We use this embedding to reinterpret the differences between the NJ and BIONJ tree building algorithms, providing one illustration of how this embedding reflects tree structures in data.
\end{abstract}

\maketitle

\section{Introduction}
Metric trees are the primary mathematical structures underlying phylogenetics, and many of its statistical analyses. A recent work by \citet{de2011} made the observation that taxon relationships representable by such trees naturally corresponded to configurations of points in Euclidean space. More specifically, given any metric tree relating a collection of taxa $X=\{x_i\}$, there is
a collection of points $P=\{\Psi(x_i)\}$ in a Euclidean space such that the distances between the points $\Psi(x_i)$ are exactly the square roots of the tree distances between the taxa $x_i$. This Euclidean point configuration represents the same information as the tree, and can offer a valuable alternative perspective, as much intuition and
many standard statistical techniques are focused on Euclidean spaces. For instance, in the above cited and a subsequent work \citep{de2012},  the Euclidean tool of principal component analysis (PCA) is applied to phylogenetic applications in a more natural way than in previous efforts.

The argument given by \citet{de2011} for this fundamental correspondence between trees and certain point configurations follows three steps: (1) a tree distance matrix is related to the covariance matrix of a diffusion model of continuous character evolution along the tree; (2) the covariance matrix is positive definite; and (3)  classical multidimensional scaling allows one to find Euclidean points realizing distances associated to such a matrix. Unfortunately, the argument given for step (1) was incomplete, so a full justification was lacking.\footnote{The gap in the argument is as follows: If $D$ denotes the $n\times n$ matrix of pairwise distances between taxa on some metric tree, and $F=I_n-\frac 1n \mathbf 1\mathbf 1^T$ where $\mathbf 1$ is a column of ones, then by multidimensional scaling theory the desired Euclidean embedding exists if and only if the ``doubly centered" matrix $H=(-1/2)FDF$ is positive semidefinite. The covariance matrix $\Sigma$ of the diffusion process on a rooted version of the tree is positive definite, and \citet{de2011} suggest that $H=\Sigma$. However, this relationship is invalid: $\Sigma$ is positive definite, while $H$ is not; $\Sigma$ depends on the root location, while $H$ does not. The correct relationship, that $H=F\Sigma F$, was not established. While the gap can be filled by proving this equality directly, our approach is simpler and more easily yields additional results.}  However, the simplicity of the conclusion suggests there should be a simpler direct explanation, not appealing to the diffusion model. The original motivation for this work was to provide one, based only in Euclidean geometry.

Our approach, however, also yields substantial new understanding, in illuminating finer geometric features of the point configurations associated to metric trees. Both the splits (bipartitions of the taxa corresponding to edges) of the tree and individual edge lengths are reflected in the point configuration in simple ways. Specifically, the splits correspond to partitions of the point configuration into orthogonal sets, and dot products of vectors between points recover the length in the tree of the common path between those taxa.  

We further show how the improvement to the neighbor joining (NJ) algorithm giving the
BIONJ algorithm for tree construction from dissimilarity data can be motivated by the point configuration. In particular, this offers a perspective on BIONJ as an iterative process involving least squares in the configuration space. Though this is perhaps not what was meant by the call by  \citet{gascuel97} for ``further exploration concerning the relationship of this [BIONJ] theory with that of generalized least-squares,''  we believe it illustrates well the value of the point configuration viewpoint.

The existence of the point configuration can also be deduced from more general theorems in distance geometry (cf.~the survey of \citet{CritchleyFichet}). However, the explicit treatment we give here seems most useful in the phylogenetic context, as the angular features of the configuration become apparent.

\section{Metric trees and point configurations}

The Euclidean point configuration corresponding to a tree is best understood by first considering a configuration
corresponding to not just the leaves of the tree, but rather all nodes, including internal ones.

Let $T$ be a metric phylogenetic tree (rooted or unrooted, not necessarily binary) with leaves uniquely labelled by the taxa in a set  $X$. Let $V=V(T)$ be the set of all $m$ nodes in the tree, with $X$ viewed as a subset of $V$.  Then $T$ has $m-1$ edges, which we arbitrarily order as $e_1, e_2,\dots,e_{m-1}$. We assume every edge length, $w(e_i)$,  is strictly positive. For any two nodes $v,u\in V$, let $P_{v,u}$ denote the oriented path (i.e., the ordered set of edges) from $v$ to $u$ in the tree $T$. The \emph{tree metric} $d:V\times V\to \mathbb R^{\ge 0}$ is then defined by

$$d(v,u)=\sum_{e\in P_{v,u}} w(e).$$

\begin{definition} Let $T$ be a metric phylogenetic tree, with the conventions above. Fix a choice of a \emph{base node} $v\in V$. Then the
 \textit{square root embedding}  of the nodes of $T$ is $\Psi_v:V\rightarrow \RR^{m-1}$, defined by  $\Psi_{v}(u) = (\alpha_1,\alpha_2,\dots,\alpha_{m-1}),$ where  

$$\alpha_i = \begin{cases}\sqrt{w(e_i)}, & \text{ if  $e_i\in P_{v,u}$},\\$0$, & \text{otherwise.}  \end{cases}$$
We refer to the $i$th coordinate in $\RR^{m-1}$ as the \textit{$e_i$-coordinate}. 
\end{definition}

\begin{figure}[h]
\begin{center}
\begin{minipage}[c]{1.6in}
\includegraphics[width=1.6in]{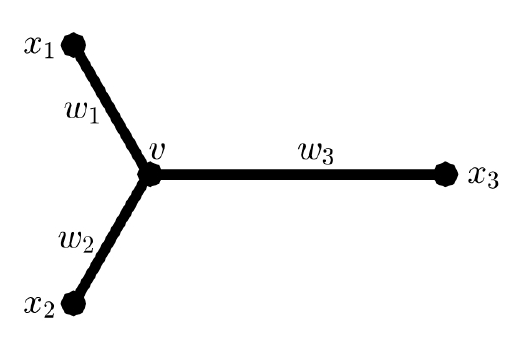}
\\ \vfill \end{minipage}
\hskip .05in
\begin{minipage}[c]{2.6in}
\includegraphics[width=2.6in]{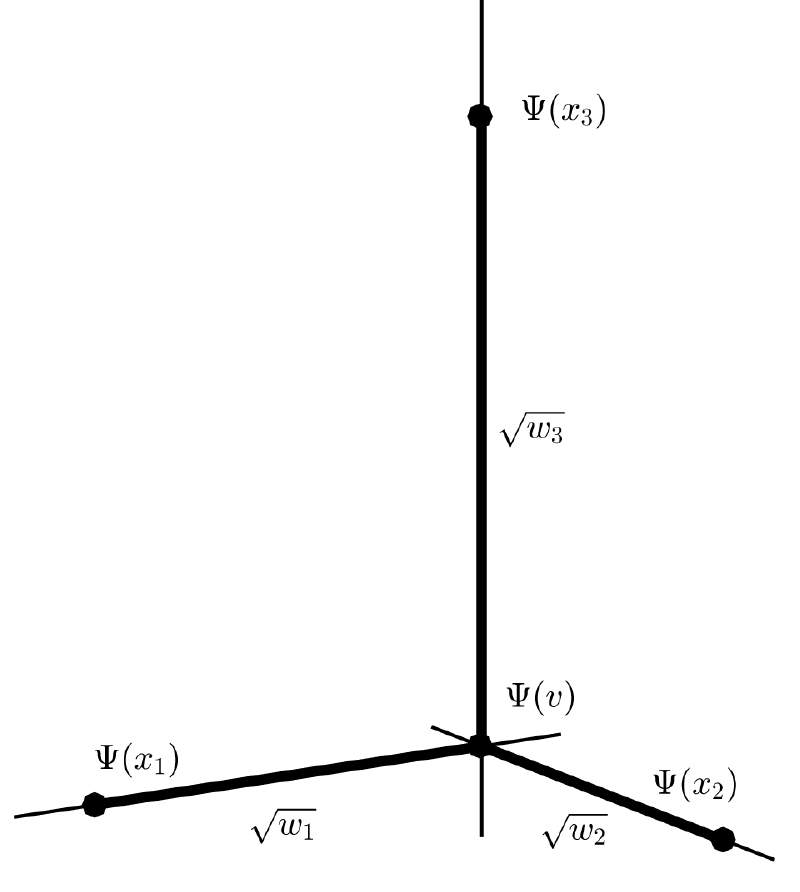}
 \end{minipage}
\hskip 1.in \phantom{X}
\end{center}
\caption{A 3-leaf unrooted tree with vertices $x_1,x_2,x_3,v$, and edge lengths $w_1,w_2,w_3$ (left) and the image of its vertices in $\RR^3$ (right) under the square root map $\Psi_{v}$.} \label{fig:3leaf}
 \end{figure}

Figure \ref{fig:3leaf} illustrates how the square root embedding sends the 4 nodes of a 3-taxon metric tree into $\RR^3$. Note that in the figure the Euclidean distance between $\Psi_{v}(x_1)=(\sqrt w_1,0, 0)$ and $\Psi_{v}(x_2)=(0,\sqrt w_2, 0)$ is 

$$\| \Psi_{v}(x_1)-\Psi_{v}(x_2)\| =\sqrt{(\sqrt{w_1})^2+(-\sqrt{w_2})^2+0^2}= \sqrt{w_1+w_2}=\sqrt{d(x_1,x_2)}.$$ Thus the fact that the Euclidean distance between $\Psi_{v}(x_1)$ and $\Psi_{v}(x_2)$ is the square root of the tree distances between $x_1$ and $x_2$ is a simple consequence of the Pythagorean Theorem.  

That this generalizes to all trees will be shown as part of Theorem \ref{thm:subpath} below. The key observation is the following.

\begin{lemma}\label{lem:posneglemma} Let $v,v_1,v_2\in V(T)$. Then  

$$\Psi_v( v_1)-\Psi_v( v_2)=( \gamma_1, \gamma_2,\dots, \gamma_{m-1})$$ $$\text{ where }  \gamma_i = \begin{cases} \phantom{-}\sqrt{w(e_i)} & \text{ when $e_i \in P_{v, v_1}$ and $e_i \notin P_{v, v_2}$,}\\
-\sqrt{w(e_i)} & \text{ when $e_i \in P_{v, v_2}$ and  $e_i \notin P_{v, v_1}$,}\\
\phantom{-}\ \ \ \ 0& \text{ otherwise.}\\
\end{cases}$$
 In particular, $ \gamma_i \neq 0$ exactly when $e_i \in P_{v_1, v_2}.$
\end{lemma}

\begin{proof}
Note that the $e$-coordinate of $\Psi_v( v_1)$ is $\sqrt{w(e)}$ only when $e \in P_{v, v_1}$, and zero otherwise. Similarly the $e$-coordinate of $\Psi_v( v_2)$ is $\sqrt{w(e)}$ only when $e \in P_{v, v_2}$ and zero otherwise. Then it immediately follows that the $e_i$-coordinate of $\Psi_v( v_1)-\Psi_v( v_2)$ is given by the stated formula for $\gamma_i$.  

The edges on the path $P_{ v_1, v_2}$ are those which lie in exactly one of $P_{v, v_1}$ and $P_{v, v_2}$. Therefore the nonzero coordinates of $\Psi_v( v_1)-\Psi_v( v_2)$ correspond to the edges on the path $P_{ v_1, v_2}$. \qed
\end{proof}

\begin{theorem}\label{thm:subpath} Let $v_1,v_2,v_3,v_4\in V(T)$. Then 
\begin{equation}
(\Psi_v(v_2)-\Psi_v(v_1))\cdot(\Psi_v(v_4)-\Psi_v(v_3)) =\pm \sum_{e\in P} w(e)\label{eq:dot}
\end{equation}
where $P=P_{v_1,v_2}\cap P_{v_3,v_4}$ is the subpath common to $P_{v_1,v_2}$ and $P_{v_3,v_4}.$ The sign is positive if the subpath is oriented in the same direction in both paths and is negative if oppositely oriented.
In particular,

\begin{equation*}
\|\Psi_v(v_2)-\Psi_v(v_1)\| = \sqrt{d(v_1,v_2)}.\end{equation*}
\end{theorem}

\begin{proof} 
By interchanging $v_3$ and $v_4$ if necessary, it is sufficient to consider the case that $P$ is oriented in the same direction in both paths. By Lemma \ref{lem:posneglemma}, any $e$-coordinate which is non-zero in both
 $\Psi_v(v_2)-\Psi_v(v_1)$ and $\Psi_v(v_4)-\Psi_v(v_3)$ arises from $e\in P$, and has absolute value $\sqrt{w(e)}$ in both. Letting $v'$ denote the node on $P$ closest to $v$, the sign of such an $e$-coordinate in both these vectors is positive if $e$ falls before $v'$ in $P$, and negative if after. Either way, the contribution to the product in equation \eqref{eq:dot} is $w(e)$, so that claim is established. 
Taking $v_1=v_2$ and $v_3=v_4$, the last claim follows.\qed
\end{proof}

Note that the right hand side of Equation \eqref{eq:dot} is independent of the base node $v$, suggesting the particular choice of $v$ is inessential.
Its precise effect is captured by the following.

\begin{proposition}\label{prop:diffv} Given a fixed ordering of the edges of $T$, and a pair of nodes $v_1,v_2$ of $T$, the maps $\Psi_{v_1}$ and $\Psi_{v_2}$ differ by coordinate reflections and translation.  More specifically there is a reflection $R$ in some coordinates of $\RR^{m-1}$ and some vector $\textbf{a}\in \RR^{m-1}$ such that   

$$\Psi_{v_1}(v) = R\Psi_{v_2}(v) + \textbf{a} \text{ for all $v$ in $V(T)$.}$$

\end{proposition}

\begin{proof}
From Lemma $\ref{lem:posneglemma}$ note 

$$\Psi_{v_1}(v)-\Psi_{v_1}(v_1) = R(\Psi_{v_2}(v) - \Psi_{v_2}(v_1))$$ where $R$ is the reflection that changes sign in $e$-coordinates with $e \in P_{v_1,v_2}$.  
Thus the above formula for $\Psi_{v_1}(v)$ holds with $\textbf{a}_{v_1,v_2} = -R\Psi_{v_2}(v_1).$\qed
\end{proof}

Since the choice of the base node only changes the image of the square root embedding by a Euclidean isometry, we generally suppress the $v$ in the notation, writing $\Psi = \Psi_v$. Note that
the definition of the embedding depends on several other arbitrary choices as well: Reordering the edges of $T$ permutes the coordinates of $\RR^{m-1}$. And if instead of using positive square roots in the definition of $\Psi_v$ we used negative ones in particular coordinates, this would only result in reflecting the image in some coordinate hyperplanes. Thus even allowing for such choices, the image is determined up to an isometry of Euclidean space.

By restricting from $\Psi(V)$ to the set $\Psi(X)$ we obtain  a configuration of points in Euclidean space corresponding only to the leaves of the tree. If $X$ has $n$ taxa, the affine span of $\Psi(X)$ is at most of dimension $n-1$, which shows that there is a point configuration in $\RR^{n-1}$ with pairwise distances equal to the square root of the tree distances between taxa. This is the sort of configuration produced by
\citet{de2011}.

\section{Further features of the point configuration}

Since a point configuration arising from a metric tree via the square root embedding reflects tree distances, it is far from arbitrary. For instance, the following theorem shows that embedded taxa always span spaces of the maximum dimension possible.

\begin{theorem} \label{thm:dimcor}Let $L\subseteq V$ be any subset of the nodes of $T$, with $n=|L|$.  Then $\Psi(L)$ spans an affine space of dimension exactly $n-1$. 
\end{theorem}

\begin{proof}
If $n=1$, then $\Psi(L)$ is one point, and so spans a space of dimension 0.

Now suppose that for any $L'\subset V$  with $1\leq |L'| < n$ that the dimension of the affine span of $\Psi(L')$ is $|L'|-1$.  Consider a set $L$ of $n$ nodes. 
Choose $w\in L$ to be any leaf of the subtree of $T$ spanned by $L$, $e$ to be the edge in that subtree containing $w$, and $v\in L$ with $v\ne w$. Let the square root embedding be given by $\Psi_v$.  

Now $L'=L \smallsetminus \{w\}$ has $n-1$ elements, so $\Psi(L')$ spans an $(n-2)$-dimensional space.  Note that the $e$-coordinate of all points in $\Psi_v(L')$ is 0.  However the $e$-coordinate of $\Psi_v(w)$ is positive, and so $\Psi_v(w)$ is not in the affine span of $\Psi_v(L')$.  Therefore the dimension of the span of $\Psi_v(L)$ is $(n-2)+1 = n-1$.\qed
\end{proof}

Since the affine span of the $|L|=n$ points in $\Psi(L)$ is $n-1$ dimensional, basic facts of Euclidean geometry imply
that any other point configuration of $n$ points with the same pairwise distances can be obtained from it by a unique Euclidean isometry of $\RR^{n-1}$ (i.e., by  rotation, reflection, and translation). The point configuration is thus essentially unique.

\begin{figure}[h]
\begin{center}
\hskip -.8in
\begin{minipage}[c]{1.5in}
\includegraphics[width=1.30in]{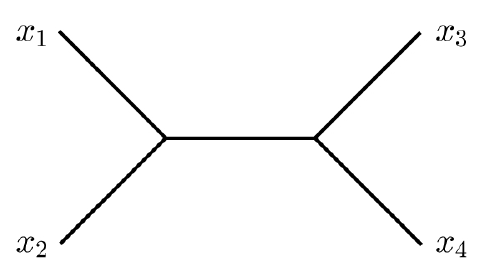}
\end{minipage}
\hskip .0in
\begin{minipage}[c]{1.in}
\includegraphics[width=3.2in]{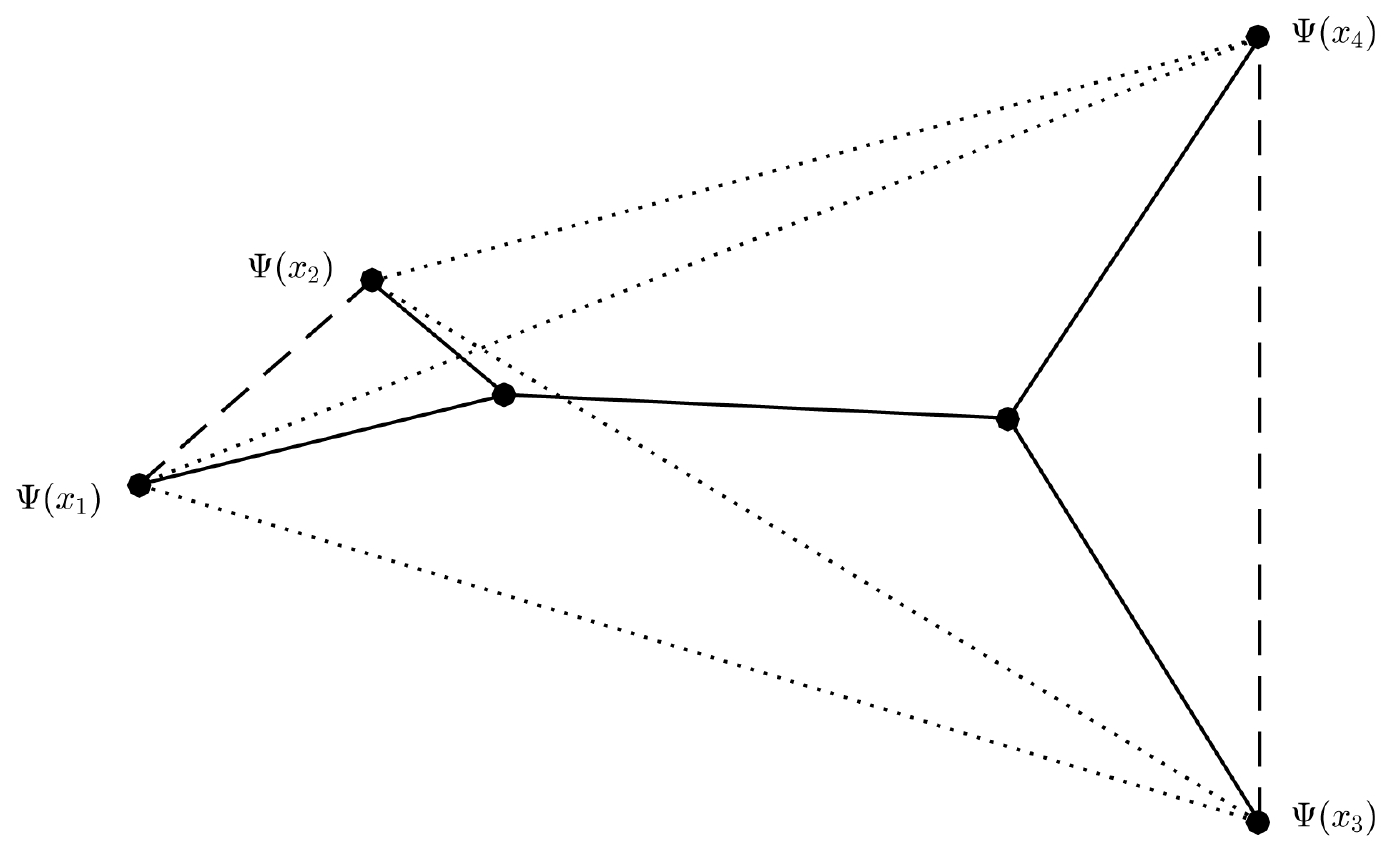}
 \end{minipage}
\hskip 1.in \phantom{X}
\end{center}
\caption{A 4-taxon tree (left) and the 3-dimensional image of its leaves under the square root embedding (right). The dotted and dashed lines are the edges of a tetrahedron, with lengths the square root of tree distances. The dashed lines connect taxa in each set of the split $\{x_1,x_2\} ~|~ \{x_3,x_4\}$, and are thus orthogonal.
The dotted lines connect taxa across the split sets, so any pair such as $\overline{\Psi(x_1)\Psi(x_3)}$ and $\overline{\Psi(x_2)\Psi(x_4)}$ form an acute angle.
All 6 vertices of the tree can be embedded only in 5-dimensional space. The tree shown inside the tetrahedron is the projection of the tree in 5-space onto the 3-space spanned by the leaves; distances along it are not those of the 5-dimensional embedding.} \label{fig:tet}
\end{figure}

While the point configuration $\psi(X)$ corresponding to the taxa on a metric phylogenetic tree was designed to encode tree distances between taxa,
these distances determine the full tree, so the configuration must contain all the information that the tree does. In particular, the configuration must reflect all splits in the tree in some geometric way.
As a motivating image, Figure \ref{fig:tet} shows a 4-taxon tree and the tetrahedron whose vertices form the corresponding configuration.  Note the edge $\overline{\Psi(x_1)\Psi(x_2)}$ is orthogonal to the edge $\overline{\Psi(x_3)\Psi(x_4)}$.  In addition, any two edges of the tetrahedron between these two orthogonal edges form an acute angle. These characteristics appear more generally, as the following shows.

\begin{theorem}\label{thm:ortho}
Let $L_1,L_2$ be two disjoint subsets of $V$. The following are equivalent:
\begin{enumerate}
\item[(a)] The minimal spanning trees of $L_1$ and $L_2$ in $T$ have no edges in common.
\item[(b)] The affine spans of $\Psi(L_1)$ and $\Psi(L_2)$ are orthogonal. 

\item[(c)] For  all $v_1,v_2\in L_1$, $w_1,w_2\in L_2$, the vectors $\Psi(v_1)-\Psi(w_1)$ and  $\Psi(v_2)-\Psi(w_2)$ form an acute or right angle.
\end{enumerate}
\end{theorem}

\begin{proof}
To see (a) and (b) are equivalent, note the affine span of $\Psi(L_i)$ is a translate of the plane generated by the vectors $\Psi(v_1)-\Psi(v_2)$ for $v_1,v_2\in L_i$.
Thus the spans are orthogonal if, and only if,

$$(\Psi(v_1)-\Psi(v_2))\cdot(\Psi(w_1)-\Psi(w_2))= 0$$
for all $v_1,v_2\in L_1$, $w_1,w_2\in L_2$.
But by Theorem \ref{thm:subpath}, this is equivalent to all paths $P_{v_1,v_2}$ and $P_{w_1,w_2}$ having no edges in common. This is in turn equivalent to the disjointness of the sets of edges in the spanning trees.

For the equivalence of (a) and (c), note that  the minimal spanning trees have an edge in common if, and only if,
there are points $v_1,v_2\in L_1$, $w_1,w_2\in L_2$ with
$P_{v_2,v_1}$ and $P_{w_1,w_2}$ having at least one an edge in common, and with the common subpath oriented in the same direction. But this is equivalent to
$P_{w_1,v_1}$ and $P_{w_2,v_2}$ having an edge in common, and common subpath oriented in the opposite direction.
 By Theorem \ref{thm:subpath}, this is exactly that

$$(\Psi(v_1)-\Psi(w_1))\cdot(\Psi(v_2)-\Psi(w_2))< 0,$$
or that these vectors form an obtuse angle.\qed
\end{proof}

One consequence of the equivalence of (a) and (c) is that no angle at a point in the configuration is obtuse: Taking $L_1=\{v_1,v_2\}$ and $L_2=\{w\}$, the minimal spanning tree for $L_2$ has no edges, and hence none in common with those of $L_1$. Thus

$$(\Psi(v_1)-\Psi(w))\cdot(\Psi(v_2)-\Psi(w))\ge 0,$$ so these vectors form an acute or right angle.

For the following, by a \emph{generalized split} of a tree $T$ we mean a bipartition of the taxa $X$ induced by deleting an edge from some binary refinement of $T$. Since a generalized split gives sets of
taxa whose minimal spanning trees contain no common edges, we immediately obtain:

\begin{corollary}\label{cor:ortho}
$X_1|X_2$ is a generalized split of the tree $T$ if, and only if, either (and hence both) of the following hold:
\begin{enumerate}
\item[(a)] The affine spans of $\Psi(X_1)$ and $\Psi(X_2)$ are orthogonal. 
\item[(b)] For  all $x_1,x_2\in X_1$, $x_3,x_4\in X_2$, the vectors $\Psi(x_1)-\Psi(x_3)$ and  $\Psi(x_2)-\Psi(x_4)$ form an acute or right angle.
\end{enumerate}
\end{corollary}

Although expressed in geometric language, this statement is essentially the same as the well-known 4-point condition: For a generalized split $X_1|X_2$ on a tree $T$ with $x_1,x_2\in X_1$, $x_3,x_4\in X_2$,

$$d(x_1, x_2) + d(x_3, x_4) \le d(x_1, x_3) + d(x_2, x_4)=d(x_1, x_4) + d(x_2, x_2).$$
Using $d(x_i,x_j)=\|\Psi(x_i)-\Psi(x_j)\|^2$ one can check that the equality here is the same as condition (a) and the inequality is the same as condition (b) of the Corollary.

\section{Point configuration from distance data}

If distances $d(x_i,x_j)$ for $n$ taxa are a tree metric, but the full tree is not yet known, then the square root map as defined above is of course not directly usable. However, since by Theorem
\ref{thm:subpath} the point configuration is proved to exist, $\{\mathbf y_i\}$ with $\|\mathbf y_i-\mathbf y_j\|=\sqrt{d(x_i,x_j)}$ can be computed by the methods of classical multidimensional scaling suggested by \citet{de2011}, and implemented in most standard statistical software. In concise form, for $n$ taxa this procedure is:
\begin{enumerate}
\item From the $n\times n$ matrix D of pairwise distances between taxa, compute the ``doubly centered" positive semidefinite symmetric matrix $H=-\frac 1 2 FDF$ where $F=I-\frac 1n\mathbf 1 \mathbf 1^T$, with $\mathbf 1$ a column vector of 1s.
\item Compute a factorization $H=X^TX$, where $X$ is a real $(n-1)\times n$ matrix. $X$ is only determined up to multiplication on the left by an $(n-1)\times(n-1)$ orthogonal matrix $Q$, since $(QX)^T(QX)=X^TQ^TQX=X^TX$.

\item
The columns of $X$ give points in $\mathbb R^{n-1}$ realizing the point configuration for the taxa. 
\end{enumerate}

The points produced by this procedure have the additional feature that their centroid is $\mathbf 0$.
The indeterminacy of $X$ up to multiplication by $Q$ reflects that distances within the configuration are preserved by rotation and reflection.
Note that this procedure only produces the point configuration for the leaves of the tree, and not for the internal nodes.

With only an estimate $\hat D \approx D$ of the true distances, 
one can still attempt to apply the same procedure to $\hat D $. If the errors are not too large, 
then $\hat H=-\frac 1 2 F\hat DF$ will also be positive semidefinite and the desired matrix factorization will still exist. This gives a point configuration approximating the true one, for which the various properties outlined in the previous section will hold only approximately. Reasoning with this approximate configuration, however, one can better understand some tree construction algorithms, as we show in the next section.

\section{Relationship to NJ and BIONJ Algorithms}

The Neighbor Joining Algorithm (NJ) \citep{SaiNei,SK88} provides the basic framework for a number of methods of building trees from approximate distance data. Given dissimilarity values $\delta_{ij}$ between all taxa $x_i,x_j\in X$ that are assumed to approximate a tree metric, it proceeds in an iterative fashion by picking a likely cherry on the unknown tree using the neighbor joining criterion, and then agglomerating those taxa. The BIONJ algorithm of \citet{gascuel97} introduced an important modification to NJ that improves performance on certain types of trees, without significantly degrading it on others. (This modification is also adopted by WEIGHBOR \citep{weighbor}, which further changes the cherry picking criterion.) Here we present the BIONJ modification in a new light, as arising naturally from the point configuration and Theorem \ref{thm:ortho}.
 
\medskip

BIONJ and NJ both use the same neighbor joining criterion, which need not be discussed here, to pick the initial two taxa to be joined. If these are taxa $x_1$ and $x_2$, then both algorithms replace the two by a single node $v$ to which $x_1,x_2$ will be joined by edges as a tree is built. They then
estimate lengths $\hat d(x_1,v)$, $\hat d(x_2,v)$ for these edges, and
calculate a dissimilarity $\delta_{vi}$ between $v$ and the remaining taxa $i=3,\dots n$ by different formulae. Having reduced the number of taxa by one, both algorithms iterate these steps.

The only difference between NJ and BIONJ with any implications for the topology of the tree to be constructed is in the calculation of  the $\delta_{vi}$ when taxa $x_1$ and $x_2$ are joined at $v$. Since \cite{gascuel94} has shown that one may add a constant (that is, a number independent of $i$) to such a formula without affecting the later behavior of the criterion used for picking cherries, we may present the formulas for these calculations most simply as follows:
Both choose
\begin{equation}\label{eq:dvi}
\delta_{vi}= \lambda \delta_{1i}+(1-\lambda) \delta_{2i},
\end{equation}
for some $0\le \lambda\le 1$.
NJ simply sets $\lambda=1/2$ while BIONJ
chooses $\lambda$ to
solve the constrained minimization problem

\begin{align}
\text{minimize } &f(\lambda)=\sum_{i=3}^n \lambda^2 \delta_{1i} +\lambda(1-\lambda)(\delta_{1i}+\delta_{2i}-\delta_{12})+
(1-\lambda)^2 \delta_{2i}\label{eq:opt}\\
\text{subject to}&\ 0\le \lambda \le 1.\notag
\end{align}
The individual terms in this sum are approximations of the variances of the $\delta_{vi}$, as derived by \citet{gascuel97} under a reasonable model of distance error.

Now suppose that the dissimilarities are sufficiently close to a tree metric that via multidimensional scaling they correspond to some Euclidean point configuration $\{\mathbf z_i\}$ with $\delta_{ij}=\|\mathbf z_i-\mathbf z_j\|^2$.
Then the formula in equation \eqref{eq:dvi} used by both NJ and BIONJ can be expressed as

$$\delta_{vi}= \lambda \|\mathbf z_1-\mathbf z_i\|^2+(1-\lambda) \|\mathbf z_2-\mathbf z_i\|^2=\| \lambda \mathbf z_1 +(1-\lambda)\mathbf z_2 -\mathbf z_i\|^2+\lambda(1-\lambda)\|\mathbf z_1-\mathbf z_2\|^2.$$
Since the last term is independent of $i$, it can be dropped to give an alternative formula which would lead to the same tree topology:

$$\tilde \delta_{vi}=\| \lambda \mathbf z_1 +(1-\lambda)\mathbf z_2 -\mathbf z_i\|^2.$$
This has a simple geometric interpretation in terms of the point configuration: The node $v$ corresponds to a point $\lambda \mathbf z_1 +(1-\lambda)\mathbf z_2$ on the line segment between the points $\mathbf z_1$ and $\mathbf z_2$ of the taxa being joined, located at proportion $\lambda$ of the way from $\mathbf z_2$ to $\mathbf z_1$. For NJ with $\lambda= 1/2$ this is the midpoint, but for BIONJ it depends on the solution of the minimization problem \eqref{eq:opt}.

The orthogonality described in part (a) of Theorem \ref{thm:ortho} implies that for a dissimilarity that is actually a tree metric, any choice of $\lambda$ will lead to the same tree topology:
Indeed, we need to check only that  $\tilde \delta_{vi}(\lambda_1)-\tilde \delta_{vi}(\lambda_2)$ is independent of $i$. Since

\begin{align*}
\tilde \delta_{vi}(\lambda_1)-\tilde \delta_{vi}(\lambda_2)&=\| \lambda_1 \mathbf z_1 +(1-\lambda_1)\mathbf z_2 -\mathbf z_i\|^2-\| \lambda_2 \mathbf z_1 +(1-\lambda_2)\mathbf z_2 -\mathbf z_i\|^2\\
&=\left ( (\lambda_1-\lambda_2) \mathbf z_1 -(\lambda_1-\lambda_2)\mathbf z_2\right )\cdot 
 \left( ( \lambda_1 +\lambda_2) \mathbf z_1 +(2-\lambda_1-\lambda_2)\mathbf z_2 -2\mathbf z_i\right )
\end{align*}
we have

$$\left ( \tilde \delta_{vi}(\lambda_1)-\tilde \delta_{vi}(\lambda_2)\right ) -\left ( \tilde \delta_{vj}(\lambda_1)-\tilde \delta_{vj}(\lambda_2)\right )=(\lambda_1-\lambda_2) (\mathbf z_1-\mathbf z_2)\cdot 2 (\mathbf z_i-\mathbf z_j)=0
$$ by the theorem.

Turning to the optimization problem \eqref{eq:opt} which determines $\lambda$ for BIONJ, the objective function $f(\lambda)$ can be expressed as

\begin{align*}
f(\lambda)&=\sum_{i=3}^n \lambda^2 \|\mathbf z_1-\mathbf z_i\|^2 +\lambda(1-\lambda)(\|\mathbf z_1-\mathbf z_i\|^2+\|\mathbf z_2-\mathbf z_i\|^2-\|\mathbf z_1-\mathbf z_2\|^2)+
(1-\lambda)^2 \|\mathbf z_2-\mathbf z_i\|^2\\
&=\sum_{i=3}^n \lambda^2 \|\mathbf z_1-\mathbf z_i\|^2 +2\lambda(1-\lambda)(\mathbf z_1-\mathbf z_i)\cdot(\mathbf z_2-\mathbf z_i)+
(1-\lambda)^2 \|\mathbf z_2-\mathbf z_i\|^2\\
&=\sum_{i=3}^n \|\lambda \mathbf z_1+(1-\lambda)\mathbf z_2-\mathbf z_i\|^2.
\end{align*}
Thus the minimization problem is to find the point on the line segment between $\mathbf z_1$ and $\mathbf z_2$ that minimizes the sum of the squares of the distances to all other points.
But this has the same minimizer as
$$\tilde f(\lambda)=\left \|\lambda \mathbf z_1+(1-\lambda)\mathbf z_2-\frac1{n-2}\sum_{i=3}^n \mathbf z_i\right \|^2$$
as the derivatives of $f$ and $\tilde f$ are the same up to a positive constant factor. Thus the calculation of $\lambda$ by BIONJ simply locates the point on the line segment between $\mathbf z_1$ and $\mathbf z_2$ that is closest to the centroid of the remaining points. 
Figure \ref{fig:BIONJ} illustrates this, through a sketch meant to represent the geometry in $(n-1)$-dimensional space.

\begin{figure}[h]
\begin{center}
\includegraphics[width=2.5in]{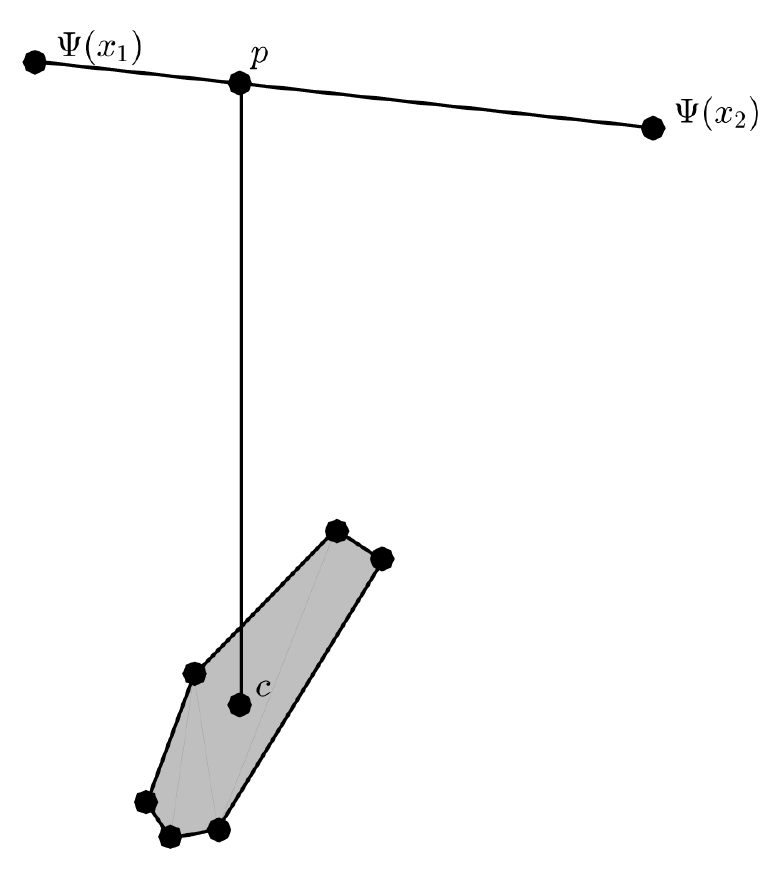}
\end{center}
\caption{Once $x_1$ and $x_2$ have been identified as forming a cherry, the BIONJ algorithm finds the point $p$ on the segment  $\overline{\Psi(x_1)\Psi(x_2)}$ that is closest to the centroid $c$ of the embeddings of the other $n-2$  taxa. In this schematic depiction the other embedded taxa are represented as the vertices of a polygon, but should be vertices of an $(n-3)$-dimensional polyhedron. 
The two line segments drawn here would then be orthogonal to all directions in that polyhedron. 
The NJ algorithm uses the midpoint of  $\overline{\Psi(x_1)\Psi(x_2)}$ in place of $p$. A variant of BIONJ discussed in the text locates $p$ on $\overline{\Psi(x_1)\Psi(x_2)}$ closest to the polyhedron.}\label{fig:BIONJ}
\end{figure}

This viewpoint suggests an extension of the innovation of BIONJ: One could choose $\lambda$ to give the point on the line segment between $\mathbf z_1$ and $\mathbf z_2$ that is closest
to the convex span of the remaining points. To put the three approaches on a common footing, for $\lambda_i,\mu_j\ge0$, $\sum_{i=1}^2 \lambda_i=1$, $\sum_{j=3}^n \mu_j=1$ consider the objective function giving the distance between a point in the span of $\mathbf z_1, \mathbf z_2$ to one in the span of the remaining points,

$$g(\boldsymbol \lambda, \boldsymbol \mu)=\left \| \sum_{i=1}^2 \lambda_i \mathbf z_i -\sum_{j=3}^n \mu_j \mathbf z_j\right \|^2.$$
Then NJ avoids minimizing this function by choosing $\boldsymbol \lambda=(1/2,1/2)$, while BIONJ simplifies the minimization in $\boldsymbol \lambda$ by first setting $\mu_j=1/(n-2)$. A third algorithm would find the minimizer $(\boldsymbol \lambda, \boldsymbol \mu)$, and use $\boldsymbol \lambda$ in
equation \eqref{eq:dvi}. If a dissimilarity is an exact tree metric, then BIONJ and the new algorithm would produce exactly the same $\boldsymbol \lambda$, due to the orthogonality of Corollary \ref{cor:ortho}; the closest point in the line segment to the convex span of the other points will also be the closest point to the centroid of those points.

But for dissimilarities including noise, the new algorithm could potentially improve performance, by allowing different weightings $\boldsymbol \mu=(\mu_j)$ for the points not in the cherry, just as BIONJ allows different weights $\boldsymbol \lambda=(\lambda_j)$ for the points in the cherry. A fourth algorithm is also possible, in which one would solve a similar optimization problem where once a split had been identified by repeated application of the usual NJ criterion, 
the objective function is built from the difference of weighted sums of points for the two split set. (In other words, rather than using only two points in the $\lambda$-sum, we use all of the original points that have been agglomerated to form these.)

The constrained optimization problems for both of these new algorithms can actually be cleanly expressed in terms of the original dissimilarities $\delta_{ij}$, so it is not necessary to calculate the $\mathbf z_i$ \citep{LayerThesis}. However, when 
we tested these algorithms on simulated DNA sequences, performance was essentially the same as that of BIONJ, as measured by topological accuracy (average Robinson-Foulds distance from correct topology, or percentage correct topology). For  some trees and sequence lengths the new algorithms might be a percentage point or two better in terms of average RF distance, but for others they were worse by similar amounts. We did not find any trees on which the new algorithms offered a substantial improvement.
Moreover, these algorithms introduce an additional computational burden of solving a quadratic minimization problem in many variables to determine $\boldsymbol \lambda$. Although there are excellent software packages for doing this, they are not as fast as using BIONJ's formula for $\boldsymbol \lambda$,
and so these approaches do not seem to be worthwhile in practice.

\section{Conclusion}

While distance methods for tree inference are seldom the first choice for data analysis, they offer significant computational advantages over full Maximum Likelihood or Bayesian analyses, and are still highly relevant to empirical work. For instance, a number of fast and statistically-consistent methods of species tree inference that proceed by first constructing a distance matrix from a collection of gene trees, and then using that to find the species tree  \citep{LiuEtAl2009,  mossel2010,  LiuEtAlMax, LiuYuUGT, JewettAndRosenberg12, ADR_STAR}.  We have shown the Euclidean point configuration associated to intertaxon dissimilarities provides an alternative viewpoint on the distance methods underlying these, and believe it may be useful for future methodological progress as well.

Though our development of the point configuration avoided reference to the diffusion model that motivated \citet{ de2011}, in the context of that model it is still natural to view it as capturing the covariance. Though we omit details, the independent contrasts introduced by
\citet{Fels85} in relation to such a model can be seen as closely tied to computing a particular set of orthogonal directions in the point configuration space, and inference of states at internal nodes of a tree reduces to linear interpolation in the configuration space.

Finally, the more detailed understanding of the Euclidean point configuration we have given should be applicable to obtaining a better understanding of uses of PCA for phylogenetic purposes, such as by \citet{de2012}. For instance, Theorem \ref{thm:dimcor}
shows that there is always some loss of information in focusing on only some principal components, while the embedding map itself allows one to investigate how tree topology and edge lengths are reflected in individual components.

\bibliographystyle{plainnat}
\bibliography{SQRembed}

\end{document}